\documentclass[runningheads]{llncs}
\usepackage[T1]{fontenc}
\usepackage{graphicx}
\usepackage{microtype}
\newif\ifarxiv

\arxivtrue 

\usepackage{sty/quiver}
\usepackage{sty/_shortcuts}
\usepackage[capitalize]{cleveref}
\usepackage[autostyle]{csquotes}
\usepackage{dutchcal}

\begin{document}
\title{Contextuality in distributed systems\thanks{Supported through an Australian Government Research Training Program Scholarship and Discovery Grant DP190102142 from the Australian Research Council (ARC).}}
%
%

\author{Nasos Evangelou-Oost\orcidID{0000-0002-8313-6127} \and
    Callum Bannister\orcidID{0000-0002-8799-054X} \and
    Ian J. Hayes\orcidID{0000-0003-3649-392X}}
%
\authorrunning{N. Evangelou-Oost et al.} 
%
\institute{The University of Queensland, St Lucia, Australia}
\maketitle              
\begin{abstract}
    We present a lattice of distributed program specifications, whose ordering represents implementability/refinement.
Specifications are modelled by families of subsets of \emph{relative} execution traces, which encode the local orderings of state transitions, rather than their absolute timing according to a global clock.
This is to overcome fundamental physical difficulties with synchronisation.
The lattice of specifications is assembled and analysed with several established mathematical tools.
Sets of nondegenerate cells of a simplicial set are used to model relative traces, presheaves model the parametrisation of these traces by a topological space of variables, and information algebras reveal novel constraints on program correctness.
The latter aspect brings the enterprise of program specification under the widening umbrella of \emph{contextual semantics} introduced by Abramsky et al.
In this model of program specifications, \emph{contextuality} manifests as a failure of a consistency criterion comparable to Lamport's definition of \emph{sequential consistency}.
The theory of information algebras also suggests efficient local computation algorithms for the verification of this criterion.
The novel constructions in this paper have been verified in the proof assistant \emph{Isabelle/HOL}.



\end{abstract}
%
%
%


\newcommand{\proj}[2]{{#1}^{\downarrow{#2}}}

\newcommand{\projtup}[2]{{#1}_{\downarrow{#2}}}

\newcommand{\lin}[1]{\angles{#1}}

\newcommand{\bang}{{!}}

\newcommand{\trace}{\cotuple}

\newcommand{\deltad}{\Delta_\mathsf{sur}}

\newcommand{\deltaplus}{\Delta_+}

\newcommand{\maxcov}{\mathsf{Cov}_\mathsf{max}}

\newcommand{\ant}[1]{\mathsf{Ant}_\mathsf{#1}}

\newcommand{\chaos}{\fun{\Theta}}

\newcommand{\chaosalg}{\cat{R}_\chaos}

\newcommand{\lambdaalg}{\fun{\Lambda}}

\newcommand{\phialg}{\fun{\Phi}}

\newcommand{\modd}[1]{\mathrm{mod\;}#1}

\newcommand{\om}{\fun{\Omega}}

\newcommand{\augsset}{\cat{Set}_{\deltaplus}}

\newcommand{\unit}{[\,]}

\newcommand{\nerve}{\fun{N}}

\newcommand{\projop}{{\downarrow}}

\section{Introduction}

Lattices of sets of traces have been successful in the field of formal methods as algebraic models for program specification and verification.
For concurrent programs, two important examples are trace models of Concurrent Kleene Algebra~\cite{DBLP:journals/jlp/HoareSMSZ16} and Concurrent Refinement Algebra~\cite{DBLP:conf/fm/HayesCMWV16}.
A core advantage of these models is that they facilitate \emph{compositional reasoning}, which mitigates the inherent difficulties in analysing the exponential proliferation of program behaviours that occur when programs run in parallel.

These models do not explicitly account for the topological structure inherent to a distributed system, which can make reasoning about local behaviour, e.g. local variable blocks, cumbersome.
Moreover, trace models for concurrency often implicitly assume a \enquote{global clock} with which all traces progress in lockstep---and this assumption limits their applicability to systems distributed over significant distances in space, due to physical (e.g. from relativistic physics) constraints on synchronisation.

In this article we describe a lattice of program specifications, that encodes the possible behaviours of a distributed system as subsets of \emph{relative} traces, as well as its configuration into independent parallel processors.

Inspirations for this work are previous uses of presheaves in concurrency~\cite{DBLP:journals/mscs/Goguen92,DBLP:journals/iandc/JoyalNW96}, the aforementioned refinement algebras~\cite{DBLP:journals/jlp/HoareSMSZ16,DBLP:conf/fm/HayesCMWV16}, topological models of concurrency such as~\cite{DBLP:conf/dagstuhl/Porter06,DBLP:books/sp/FajstrupGHMR16}, and the diverse applications of valuation and information algebras~\cite{pouly2012generic}, their relationship to sheaves, and their associated notion of \emph{contextuality}~\cite{MR4024458,DBLP:phd/ethos/Caru19}.

In \cref{sec.traces} we introduce the notion of a \emph{relative trace}, which is a chain in a proset (preordered set) of states.
Such relative traces are assigned to each variable $v \in V$ in a distributed system, and are modelled as a subset of the cartesian product $\prod_{v \in V}\om_v$ of the prosets $\om_v$ of states for each variable.
The variables $v \in V$ are topologised in a space $(V,\cat{D})$ representing the physical configuration of the variables, and the inclusion ordering of open sets in this space induces a restriction action on the traces.
These traces are called \emph{relative} because this restriction action does not preserve the absolute timing of their states, but only the ordering of state transitions.

In \cref{sec.specs}, we introduce the notion of a \emph{specification}, which is a pair $(\fun{A}, \cat{U})$ where $\fun{A}$ is a presheaf on $\cat{D}$ whose values are subsets of possible relative traces, and $\cat{U}$ is a maximal cover of $(V,\cat{D})$, representing the distribution of the specification into independent asynchronous components.
Moreover, we explain that such specifications form a lattice $(\lambdaalg, \preceq)$ whose ordering represents the refinement relation between specifications.

In \cref{sec.info}, we define an \emph{information algebra} and some associated constructs.
We show how the lattice of specifications $\lambdaalg$ corresponds to a particular ordered, adjoint information algebra.
This associated information algebra permits the definition of \emph{local} and \emph{global consistency} for specifications, which we introduce in \cref{sec.consistency}.

In \cref{sec.dining}, we show how these consistency criteria can arise in a classical scenario in distributed systems, namely the \emph{dining philosophers}.

In the following we assume familiarity with the basic definitions of order theory and category theory, e.g. of a proset, a lattice, a category, a functor, a natural transformation, etc.
Possibly less familiar structures---topological spaces, presheaves, and (augmented) simplicial sets---are briefly reviewed.

Several of the constructions and proofs in this paper have been formalised in the \emph{Isabelle/HOL} proof assistant.\footnote{\url{https://github.com/onomatic/ramics23-proofs}}
\section{Relative traces}
\label{sec.traces}

A \defn{topological space} is an abstract model of a geometric space, but without built-in notions of angle, distance, curvature, etc.
It is formalised set-theoretically as a set $X$ of \defn{points}, and a set $\cat{T} \subseteq \fun{P} X$ of subsets of $X$, called \defn{open sets}, which are closed under unions and finite intersections\footnote{Consequently, $\cat{T}$ contains at least $\emptyset$ and $X$, being the union and intersection respectively of an empty family of open sets.}.
Roughly speaking, the open sets measure proximity, where nearby points occupy many open sets in common.

Let $V$ be a finite\footnote{Finiteness is not crucial, but it simplifies our presentation, and in real-world examples finiteness is a realistic assumption.} set of \defn{variables} of a distributed computer system and $\cat{D}$ a topology on $V$.
We call the open sets of $\cat{D}$ \defn{domains}.
We consider that a single computer is also a distributed system (as in~\cite{DBLP:journals/cacm/Lamport78}); then $V$ could be the set of memory locations over its CPU caches, RAM modules, hard disks, etc., and the topology $\cat{D}$ encodes the connectivity between these parts.
Or, $V$ could be the set of memory locations over a distributed database comprised of many individual computers, and $\cat{D}$ represents the network topology of this distributed system.

\subsubsection{The frame functor $\om$.}

To each variable $v \in V$ we associate a \emph{nonempty} proset $\om_v$ of \defn{states} whose order represents reachability or causality, and we extend this assignment to open sets $U \in \cat{D}$ by setting
\begin{equation}
    \om U \defeq \prod_{v \in U} \om_v
\end{equation}
where the right-hand side is a cartesian product of prosets (for which the ordering is given componentwise).
Moreover, each inclusion of open sets $U \subseteq U'$ induces a function $\om U' \to \om U$ by projection of tuples (i.e. function restriction).
Such restriction maps $\om U' \to \om U$ have the effect of discarding information involving variables outside $U$.

We evidently have that the function induced by $U \subseteq U$ is the identity, and if $U'' \subseteq U' \subseteq U$ then the functions $\om U \to \om U''$ and the composite of $\om U' \to \om U''$ and $\om U \to \om U'$ are equal---it does not matter whether we restrict tuples immediately to $U''$, or first restrict to $U'$ and then restrict to $U''$.
These assignments and properties are summarised in saying that $\om : \op{\cat{D}} \to \cat{Pro}$ is a contravariant functor from the posetal category $\cat{D}$ to the category of prosets, dubbed the \defn{frame functor}.

\subsubsection{The augmented simplicial nerve functor $\fun{N}$.}

A \defn{presheaf} is a contravariant functor valued in sets.
An \defn{augmented simplicial set}\footnote{An ordinary \defn{simplicial set} is a presheaf on the full subcategory $\Delta \subset \deltaplus$ consisting of only \emph{nonempty} linearly ordered posets.} is a presheaf $S : \op{\deltaplus} \to \cat{Set}$ whose domain is the \defn{augmented simplex category} $\deltaplus$.
This is the category with
\begin{itemize}
    \item objects as linearly ordered posets $\lin{n} \defeq \setc{i \in \Nn}{0 \leq i \leq n}$ for integers $n \geq 0$, as well as $\lin{-1} \defeq \emptyset$,
    \item morphisms as weakly monotone functions $f : \lin{n} \to \lin{m}$, i.e. those satisfying $i \leq j \implies fi \leq fj$.
\end{itemize}

Augmented simplicial sets are the objects of a category, denoted by $\augsset$, with natural transformations as morphisms.
It is conventional to write the application $S \lin{n}$ of an augmented simplicial set $S$ on $\lin{n} \in \deltaplus$ as $S_n$, and the application $Sf$ of $S$ on a morphism $f \in \deltaplus$ as $\contra{f}$.

For an augmented simplicial set $S \in \augsset$, an element $x \in S_n$ is called an \defn{$n$-cell of $S$}, or a \defn{cell of $S$ (of degree $n$)}.
The cell $x$ is \defn{degenerate} if there exists a non-injective function $f\in \deltaplus$ and a cell $y$ of $S$ with $x = \contra{f} y$.
Note that any cell of degree $-1$ is nondegenerate.

To any proset $\cat{P}$ we can produce an augmented simplicial set $\nerve\cat{P}$, called the \defn{augmented simplicial nerve of $\cat{P}$,} whose action on objects $\lin{n}$ for $n \in \Nn \cup \set{-1}$ is given by
\begin{equation}
    \sub{\nerve\cat{P}}{n} \defeq \homset{\cat{Pro}}{\lin{n}}{\cat{P}}
\end{equation}
where $ \homset{\cat{Pro}}{\lin{n}}{\cat{P}}$ is the set of monotone functions from $\lin{n}$ to $\cat{P}$, or equivalently the set of all \defn{chains of length $n$} in $\cat{P}$, including the empty chain, denoted $\unit$, which is said to have length $-1$.
Given a morphism $f : \cat{P} \to \cat{P}'$ in $\cat{Pro}$, we define $\nerve f : \nerve\cat{P} \to \nerve\cat{P}'$ as postcomposition with $f$, i.e.
\begin{equation}
    (x : \lin{n} \to \cat{P}) \mapsto (f \circ x : \lin{n} \to \cat{P}')
\end{equation}
This evidently defines a functor $\nerve : \cat{Pro} \to \augsset$.
Note that for $\cat{P} \in \cat{Pro}$ a proset, a nondegenerate $n$-cell of $\nerve\cat{P}$ is a chain of length $n$ in $\cat{P}$ \emph{that contains no repeated adjacent elements}.
In particular, the empty chain $\unit$ is nondegenerate.

For each domain $U \in \cat{D}$, $(\nerve \circ \om)U$ is an augmented simplicial set, such that for each $n \in \Nn \cup \set{-1}$, $\sub{(\nerve \circ \om)U}{n}$ is the set of all possible sequences of states in $\om U$, and the functions $\sub{(\nerve \circ \om)U}{n} \to \sub{(\nerve \circ \om)U}{m}$ are generated by \emph{mumbling} and \emph{stuttering} maps on traces~\cite{DBLP:journals/iandc/Brookes96} (i.e. maps that omit or repeat elements of a sequence, respectively).

\subsubsection{The nondegenerate cells functor $\fun{D}$.}

A basic result in the theory of simplicial sets is the following:
\begin{lemma}[{Eilenberg-Zilber~\cite[II.3.1, pp. 26-27]{MR0210125}}\footnote{The cited result is stated for ordinary simplicial sets, but it clearly also applies to augmented simplicial sets as there are no degeneracies of a cell of degree $-1$.}]
    \label{ch1.lem:eilenberg-zilber}
    For each cell $x$ of an augmented simplicial set $S$ there exists a unique nondegenerate cell $x'$ such that there exists a unique surjection $f_x \in \deltaplus$ with $x= \contra{f_x}x'$.
\end{lemma}

For any augmented simplicial set $S$ we can produce a plain set $\fun{D} S$ consisting of only the nondegenerate cells of $S$ (as in~\cref{ch1.lem:eilenberg-zilber}).
Moreover, we can extend this assignment to morphisms of $\augsset$,
\begin{align}
    \fun{D} & : \augsset \to \cat{Set}                        \\
    S       & \mapsto \text{set of nondegenerate cells of } S \\
    \alpha  & \mapsto (x \mapsto (\alpha x)')
\end{align}

In other words, the action of $\fun{D}$ on morphisms of augmented simplicial sets $\alpha : S \to S'$ gives a function $\fun{D}\alpha : \fun{D}S \to \fun{D}S'$ that sends a nondegenerate cell $x$ of $\fun{D}S$ to the nondegenerate cell $(\alpha x)'$ of $S'$ that generates $\alpha x$, whose existence and uniqueness are assured by~\cref{ch1.lem:eilenberg-zilber}.
These assignments assemble to a functor.

\begin{lemma}\ifarxiv [proof in~\cref{lem.nondegen_functor_proof}] \fi
    \label{lem.nondegen_functor}
    $\fun{D}$ is a functor.
\end{lemma}

Postcomposing $\nerve \circ \om$ by $\fun{D}$, we obtain the presheaf
\begin{equation}
    \chaos \defeq \fun{D} \circ \nerve \circ \om : \op{\cat{D}} \to \cat{Set}
\end{equation}
that we call \defn{chaos}.

\begin{definition}[relative trace]
    For $U \in \cat{D}$, an element $t \in \chaos U$ is a \defn{$U$-relative trace}, or for short, a \defn{$U$-trace} or simply a \defn{trace}.
\end{definition}

\note{
    We use the standard shorthand for restriction maps of a presheaf $F : \opcat{\cat{T}} \to \cat{Set}$ (e.g. $F = \chaos$) on a topological space $(X,\cat{T})$; $\restrict{t}{U} \defeq (F i)t$ for $t \in F U'$ and $U,U' \in \cat{D}$ with $U' \subseteq U$ and $i$ the unique morphism $i : U \to U'$.
}

\note{
    \label{note.matrix}
    By fixing a linear ordering of $V$, a $U$-trace $t \in \chaos U$ for some $U \in \cat{D}$ can be represented by a unique matrix with rows labelled in increasing order by the $v \in U$, and columns indexed by \enquote{time}, with the property that adjacent columns of the matrix are always distinct (the empty trace $\unit \in \chaos U$ is therefore represented by the unique matrix with $\abs{U}$ rows and zero columns.).
    This property could be restated as saying that the traces $t \in \chaos U$ do not contain \emph{stutterings}~\cite{DBLP:journals/iandc/Brookes96}.
    The qualifier \enquote{relative} applied to \enquote{traces} emphasises the latter property, which entails that a trace only records the relative ordering of its states, rather than their absolute timing according to an implied \enquote{global clock}.
    This is illustrated in the example below.
}
\begin{example}
    \label{ex.trace}
    Let $(V,\fun{P} V)$ be the discrete space on a set $V \defeq \set{a,b}$ of variables, and $\om_a \defeq \set{a_0, a_1}$ and $\om_b \defeq \set{b_0, b_1}$ their corresponding prosets of states, both with the \emph{total ordering}, i.e. for which all pairs of elements are related.
    The trace
    \begin{equation}
        t =
        \begin{bmatrix}
            a_0 & a_0 & a_1 \\
            b_0 & b_1 & b_1
        \end{bmatrix}
        \in \chaos V
    \end{equation}
    informally corresponds to an ordered set of observations of the system, where the state transition $a_0 \leadsto a_1$ is observed after the transition $b_0 \leadsto b_1$.
    On the other hand, the trace
    \begin{equation}
        t' =
        \begin{bmatrix}
            a_0 & a_1 \\
            b_0 & b_1
        \end{bmatrix}
        \in \chaos V
    \end{equation}
    corresponds to a discretely ordered set of observations where neither $a_0 \leadsto a_1$ is observed before $b_0 \leadsto b_1$ or vice versa (although we refrain from saying they are simultaneous/synchronous).
    Moreover, we have
    \begin{equation}
        \restrict{t}{\set{a}} = \restrict{t'}{\set{a}} =
        \begin{bmatrix}
            a_0 & a_1
        \end{bmatrix}
        \in \chaos \set{a}
    \end{equation}
    It is for this reason we refer to such traces as \emph{relative}, because only the relative ordering of states is preserved under restriction maps: the transition $a_0 \leadsto a_1$ in $t$ was at \enquote{$\code{time}=2$} but in $\restrict{t}{\set{a}}$ it is at \enquote{$\code{time}=1$}.
\end{example}

\section{Specifications}
\label{sec.specs}

A \defn{cover} of a topological space $(X,\cat{T})$ is a family $\cat{U} = \set{U_i}_i$ of open sets whose union $\bigcup_i U_i$ equals $X$.
A \defn{maximal cover} is a cover that is also an \defn{antichain}, meaning $U_i \subseteq U_j$ if and only if $i = j$.
We call a maximal cover of $(V,\cat{D})$ a \defn{context}, where each $U_i \in \cat{U}$ is a \defn{domain} (of an independent process of the distributed system).

A \defn{subpresheaf} $\fun{G}$ of a presheaf $\fun{F} : \op{\cat{C}} \to \cat{Set}$ on a category $\cat{C}$, written $\fun{G} \subseteq \fun{F}$, is a family of subsets $\fun{G} X \subseteq \fun{F} X$ for each object $X \in \cat{C}$ that assemble to a presheaf, where $\fun{G} $ inherits the action of $\fun{F}$ on morphisms $X' \to X$ (i.e. by function restriction).

\begin{definition}[specification]
    A pair $(\fun{A} ,\cat{U})$, where $\fun{A} \subseteq \chaos$ is a subpresheaf of $\chaos$, and $\cat{U} = \set{U_i}_i$ is a context, is called a \defn{specification}.
\end{definition}

The first factor of a specification records the possible relative execution traces of a distributed system, and the second defines the domains of the independent asynchronous processes that make up the system.
Specifications are partially ordered, where
\begin{equation}
    (\fun{A} ,\cat{U}) \preceq (\fun{B} , \cat{W})
\end{equation}
if and only if both $\fun{A} \subseteq \fun{B}$ and $\cat{U}$ \defn{refines} $\cat{W}$, meaning every open $U \in \cat{U}$ is contained in some $W \in \cat{W}$.
This ordering represents \emph{implementation} (or \emph{refinement}) of specifications: $(\fun{A},\cat{U}) \preceq (\fun{B}, \cat{W})$ means the left-hand side \emph{implements} (or \emph{refines}) the right-hand side.

\begin{example}
    Let the set of variables $V \defeq \set{a,b}$ equipped with the discrete topology, and let the subpresheaf $\fun{A}$ be defined by $\fun{A} \emptyset \defeq \set{\unit}$, $\fun{A} \set{a} \defeq \set{ \begin{bmatrix}
                a_0 & a_1
            \end{bmatrix}}$, and $\fun{A} \set{b} \defeq \fun{A} V \defeq \emptyset$, then it holds that $\fun{A}$ with the context $\set{\set{a}, \set{b}}$ refines chaos with the \defn{trivial context} $\set{V}$, i.e.
    $
        (\fun{A}, \set{\set{a},\set{b}}) \preceq (\chaos, \set{V})
    $.
\end{example}
Note that refinement in the first factor represents \emph{reduction of nondeterminism}, whereas in the second factor it is \emph{increase of parallelism}.

The subpresheaves of $\chaos$ form a complete distributive lattice\footnote{Actually, a \emph{complete bi-Heyting algebra}~\cite[Cor. 9.1.13]{reyes2004generic}.} $\subs{\chaos}$~\cite[\S III.8 Prop. 1]{MR1300636}\footnote{The cited result is stated more generally for the lattice of sub\emph{sheaves} of a given \emph{sheaf} over a \emph{site}.  Here we take the \emph{trivial site}, over which sheaves are equivalent to presheaves.} with meet and join given by pointwise intersection and union, i.e.
\begin{equation}
    (\fun{A} \cap \fun{B})U  = \fun{A} U \cap \fun{B} U,\qquad
    (\fun{A} \cup \fun{B})U  = \fun{A} U \cup \fun{B} U
\end{equation}

\begin{theorem}\ifarxiv [proof in~\cref{thm.covers_proof}] \fi
    \label{thm.covers}
    The set of maximal covers $\maxcov{\cat{T}}$ of a space $(X,\cat{T})$ with the refinement ordering $\leq$ described above, forms a complete distributive lattice, with meet and join given for all $\cat{U}, \cat{W} \in \maxcov{\cat{T}}$ respectively by
    \begin{align}
        \label{eq.antichain_meet}
        \cat{U} \land \cat{W} & = \setc{U \in L_\cat{U} \cap L_\cat{W}}{\not\te{ V \in L_\cat{U} \cap L_\cat{W}}{ U \subset V}} \\
        \label{eq.antichain_join}
        \cat{U} \lor \cat{W}  & = \setc{U \in \cat{U} \cup \cat{W}}{\not\te{V \in \cat{U} \cup \cat{W}}{U \subset V}}
    \end{align}
    where
    $
        L_\cat{U} \defeq
        \setc{U \in \cat{T}}{\te{V \in \cat{U}}{U \subseteq V}}
    $.
\end{theorem}

The set of specifications $\lambdaalg$ is then defined as the cartesian product of the distributive lattices,
\begin{equation}
    \lambdaalg \defeq
    \subs{\chaos} \times \maxcov{\cat{D}}
\end{equation}
where $\maxcov{\cat{D}}$ is the lattice of maximal covers of $(V, \cat{D})$, and this is again a complete distributive lattice~\cite[p. 12]{MR598630}, with meet and joint defined pointwise, i.e.
\begin{align}
    (\fun{A}, \cat{U}) \land (\fun{B}, \cat{W}) & = (\fun{A} \cap \fun{B}, \cat{U} \land \cat{W}) \\
    (\fun{A}, \cat{U}) \lor (\fun{B}, \cat{W})  & = (\fun{A} \cup \fun{B}, \cat{U} \lor \cat{W})
\end{align}
for $(\fun{A}, \cat{U}), (\fun{B}, \cat{W}) \in \lambdaalg$.
Informally, the meet of two specifications in $(\lambdaalg, \preceq)$ is the specification that contains all behaviours common to both while increasing parallelism to the minimal extent, whereas the join of two specifications is the specification containing the union of their behaviours while decreasing parallelism to the maximal extent.

The lattice $(\lambdaalg, \preceq)$ has as top element
$
    \top = (\chaos, \set{V}),
$
and as bottom element
$
    \bot = (\emptyset, \bigwedge \maxcov{\cat{D}}),
$
where $\emptyset$ is the constant functor with $\emptyset U = \emptyset$ and with $\emptyset i = 1_\emptyset$ the identity function on the empty set, for all open sets $U \in \cat{D}$ and inclusions $i : U \subseteq U'$.
Note that the meet $\bigwedge \maxcov{\cat{D}}$ over all contexts in $\maxcov{\cat{D}}$ exists because $\cat{D}$ is finite; this is the \defn{finest} context of $\cat{D}$.

Our goal in the next subsection is to show that the specifications of $(\lambdaalg, \preceq)$ can be profitably analysed through a structure known as an \emph{information algebra}.

\section{Information algebras}

\label{sec.info}

An \emph{information algebra} is an algebraic structure modelling information parameterised over a lattice of domains, together with \emph{combination} and \emph{projection} operators.
These specialise the \emph{valuation algebras} introduced by Shenoy~\cite{DBLP:journals/ijar/Shenoy89}.
Our use of information algebras is motivated by the theory of \emph{contextual semantics} developed in~\cite{MR4024458,DBLP:phd/ethos/Caru19}.
However,~\cite{MR4024458,DBLP:phd/ethos/Caru19} assume a discrete topology, whereas we prefer to allow arbitrary finite topological spaces to make a closer connection between the mathematical model of a distributed system and its physical topological configuration.
Therefore, in the following we mildly generalise the definitions and results of~\cite{MR4024458,DBLP:phd/ethos/Caru19} to arbitrary finite topological spaces.

\begin{definition}[information algebra]
    Let $(X,\cat{T})$ be a topological space over a finite set of variables $X$.
    An \defn{information algebra} over $\cat{T}$ is a quintuple $(\phialg, \cat{T}, \dd, \projop, \otimes)$, where $\phialg$ is a set, $\dd$ a function, $\otimes$ a binary operation, and $\projop$ a partially defined operation,
    \begin{enumerate}
        \item \emph{Labelling}: $\dd : \phialg \to \cat{T}, \phi \mapsto \dd \phi$,
        \item \emph{Projection}: $\projop : \phialg \times \cat{T} \to \phialg, (\phi, U) \mapsto \proj{\phi}{U}$, defined for all $U \subseteq \dd \phi$,
        \item \emph{Combination}: $- \otimes - : \phialg \times \phialg \to \phialg, (\phi,\psi) \mapsto \phi \otimes \psi$,
    \end{enumerate}
    such that the following properties (explained below) hold,
    where for $U \in \cat{T}$, $\phialg_U \defeq \setc{\phi \in \phialg}{\dd \phi = U}$, and where $\phi, \psi \in \phialg$:
    \begin{enumerate}
        \item[(I1)] \emph{Commutative semigroup}: $(\phialg,\otimes)$ is associative and commutative.
        \item[(I2)] \emph{Projection}: given $U \subseteq \dd \phi$,
            $
                \dd (\proj{\phi}{U}) = U
            $.
        \item[(I3)] \emph{Transitivity}: given $W \subseteq U \subseteq \dd \phi$, $
                \proj{
                    (\proj{\phi}{U})
                }{W} = \proj{\phi}{W}
            $.
        \item[(I4)] \emph{Domain}:
            $
                \proj{\phi}{\dd \phi} = \phi
            $.
        \item[(I5)] \emph{Labelling}:
            $
                \dd (\phi \otimes \psi) = \dd \phi \cup \dd \psi
            $.
        \item[(I6)] \label{axiom.combination} \emph{Combination}: for $U \defeq \dd \phi $, $W \defeq \dd \psi$ and $Q \in \cat{T}$ such that $U \subseteq Q \subseteq U \cup W$, we have
            $
                \proj{(\phi \otimes \psi)}{Q}
                = \phi \otimes \proj{\psi}{Q \cap W}
            $.
        \item[(I7)] \emph{Neutrality}: for each $U \in \cat{T}$, there exists a \emph{neutral element} $1_U \in \phialg_U$ such that $
                \phi \otimes 1_U
                = 1_U \otimes \phi
                = \phi
            $
            for all $\phi \in \phialg_U$.
            Moreover, these neutral elements satisfy
            $
                1_U \otimes 1_W
                = 1_{U \cup W}
            $
            for all $U, W \in \cat{T}$.
        \item[(I8)] \emph{Nullity}: for each $U \in \cat{T}$, there exists a \emph{null element} $0_U \in \phialg_U$ such that
            $
                \phi \otimes 0_U
                = 0_U \otimes \phi
                = 0_U
            $.
            Moreover, for all $U, W \in \cat{T}$ with $W \subseteq U$ and $\phi \in \phialg_U$, these null elements satisfy
            $
                \proj{\phi}{W} = 0_W \iff \phi = 0_U
            $.
        \item[(I9)] \emph{Idempotence}:
            For all $U \subseteq \dd \phi$, it holds that
            $
                \phi \otimes \proj{\phi}{U}
                = \phi
            $.
    \end{enumerate}
\end{definition}

The elements $\phi \in \phialg$ of an information algebra $(\phialg, \cat{T}, \dd, \projop, \otimes)$ are called \defn{valuations}.
An element $U \in \cat{D}$ is called a \defn{domain}.
The \defn{domain of a valuation} $\phi$ is the set $ \dd \phi \in \cat{D}$.

Some explanation for these axioms may be helpful.
Axiom (I1) says the order in which information is combined is irrelevant.
Axioms (I2)--(I4) essentially say that the triple $(\phialg,\dd, \projop)$ defines the structure of a presheaf (see~\cref{note.prealgebra} below).
(I5) is clear.
(I6) is the subtlest of the axioms; it says that to add a new piece of information, we can first strip its irrelevant parts.
This turns out to be crucial in developing efficient computational algorithms~\cite{DBLP:books/daglib/0008195,pouly2012generic}.
(I7) posits neutral elements, that contain \enquote{irrelevant} information, in the sense that combining with them adds nothing new, whereas (I8) posits null elements of \enquote{destructive} or \enquote{contradictory} information, that \enquote{corrupt} any other information combined with them.
(I9) distinguishes information algebras from their more general cousins, \emph{valuation algebras}, and is \enquote{the signature axiom of qualitative or logical, rather than quantitative, e.g. probabilistic, information. It says that counting how many times we have a piece of information is irrelevant}~\cite{MR4024458}.

\note{
    \label{note.prealgebra}
    Any information algebra $(\phialg, \cat{T}, \dd, \projop, \otimes)$ determines a presheaf $\phialg : \op{\cat{T}} \to \cat{Set}$, defined on objects $U \in \cat{T}$ by $\phialg U \defeq \phialg_U$, and such that if $W \subseteq U$ we have an action of restriction defined by projection, i.e. $\restrict{\phi}{W} \defeq \proj{\phi}{W}$, for $\phi \in \phialg U$.
    This presheaf is called the \defn{prealgebra associated to the information algebra $\phialg$}.
    We sometimes use this without mention.
}

Information algebras may often be enriched with a partial ordering on valuations, enabling the relative quantification of their information content~\cite{DBLP:journals/ijar/Haenni04}.

\begin{definition}[ordered information algebra]
    \label{def.ordered_info_alg}
    Let $\phialg$ be an information algebra on a space of variables $(X,\cat{T})$.
    Then $(\phialg,\cat{T},\dd, \projop, \otimes,\leq)$ is an \defn{ordered information algebra} if and only if $\leq$ is a partial order on $\phialg$ such that the following axioms hold:
    \begin{enumerate}
        \item[(O1)] \emph{Partial order}: for all $\phi, \psi \in \phialg$, $\phi \leq \psi$ implies $\dd \phi = \dd \psi$.
            Moreover, for every $U \in \cat{T}$ and $\Psi \subseteq \phialg_U$, the infimum $\Inf{\Psi}$ exists.
        \item[(O2)] \emph{Null element}: for all $U \in \cat{T}$, we have $\Inf{\phialg_U} = 0_U$.
        \item[(O3)] \emph{Monotonicity of combination}: for all $\phi_1, \phi_2, \psi_1, \psi_2 \in \Psi$ such that $\phi_1 \leq \phi_2$ and $\psi_1 \leq \psi_2$ we have $\phi_1 \otimes \psi_1 \leq \phi_2 \otimes \psi_2$.
        \item[(O4)] \emph{Monotonicity of projection}: for all $\phi, \psi \in \phialg$, if $\phi \leq \psi$, then $\proj{\phi}{U} \leq \proj{\psi}{U}$, for all $U \subseteq \dd \phi = \dd \psi$.
    \end{enumerate}
\end{definition}

Generically, we can interpret $\phi \leq \psi$ for $\phi,\psi \in \phialg$ as meaning $\psi$ is \emph{less}\footnote{Note that the ordering is in the \enquote{wrong} sense; this is so it corresponds to subset inclusion in~\cref{thm.tuple_to_info} below.
    A different, \emph{canonical ordering}, is used in~\cite{DBLP:books/daglib/0008195}, defined $\phi \leq_\mathsf{can} \psi \iff \phi \otimes \psi = \psi$ for all $\phi,\psi \in \phialg$.
    Actually, we have $\phi \leq \psi \implies \psi \leq_\mathsf{can} \phi$.
} informative than $\phi$.
Null elements $0_U$ for $U \in \cat{T}$ represent \emph{over-constrained}, or \emph{contradictory} information involving the variables $U$.

\subsubsection{Tuple system structure.}

Our goal now is to show that the lattice $(\lambdaalg, \preceq)$ of specifications introduced in~\cref{sec.specs} is naturally associated to a particular ordered information algebra.
To this end, we first introduce an auxiliary construction known as a \emph{tuple system}, which generalises the idea of a parameterised set of cartesian (i.e. ordinary) tuples.

\begin{definition}[tuple system]
    A \defn{tuple system} over a lattice $\cat{L}$ is a quadruple $(\fun{T},\cat{L},\dd,\projop)$, where $\fun{T}$ is a set, $\dd : \fun{T} \to \cat{L}$ a function, and $\projop : \fun{T} \times \cat{L} \to \fun{T}$ a partially defined operation, such that $\projtup{x}{U}$ is defined only when $U \leq \dd x$, and which satisfy the following axioms: for $x,y \in \fun{T}$ and $U,W \in \cat{L}$,
    \begin{enumerate}
        \item[(T1)] if $U \leq \dd x$ then $\dd (\projtup{x}{U}) = U$,
        \item[(T2)] if $W \subseteq U \subseteq \dd x$ then $\projtup{(\projtup{x}{U})}{W} = \projtup{x}{W}$,
        \item[(T3)] if $\dd x = U$ then $\projtup{x}{U} = x$,
        \item[(T4)] for $U \defeq \dd x $, $ W \defeq \dd y$, if $\projtup{x}{U \land W} = \projtup{y}{U \land W}$, then there exists $z \in \fun{T}$ such that $\dd z = U \lor W$, $\projtup{z}{U} = x$ and $\projtup{z}{W} = y$,
        \item[(T5)] for $\dd x = U$ and $U \leq W$, there exists $y \in \fun{T}$ such that $\dd y = W$ and $\projtup{y}{U} = x$.
    \end{enumerate}
\end{definition}

\note{Similar to \cref{note.prealgebra}, axioms (T1)--(T3) imply that $\fun{T}$ is associated to a presheaf $\fun{T} : \op{\cat{L}} \to \cat{Set}$ in an evident way.
    Also, any information algebra defines a tuple system, with the same domain and projection operations~\cite[Lemma 6.11, p. 170]{DBLP:books/daglib/0008195}.
}

\begin{theorem}\ifarxiv [proof in~\cref{thm.concurrency_tuple_proof}] \fi
    \label{thm.concurrency_tuple}
    The set $\chaos \defeq \coprod_{U \in \cat{D}} \chaos U$ with $\dd \defeq \pi_1 : \chaos \to \cat{D}$ the first projection from the disjoint union $\chaos$, i.e. $(U,\phi) \mapsto U$, and $\projop$ defined by restriction relative to the presheaf $\chaos$, i.e.
    $
        \projtup{x}{U} \defeq \restrict{x}{U} \defeq (\chaos{i}) x,
    $
    where $i$ is the inclusion $i : U \into \dd x$,
    defines a tuple system over the space $\cat{D}$ of domains.
\end{theorem}

For a tuple system $(\fun{T}, \cat{L})$ and $U \in \cat{L}$, a subset $A \subseteq \fun{T}_U \defeq \setc{x \in \fun{T}}{\dd x = U}$ is called a \emph{relation}\footnote{In~\cite{DBLP:phd/ethos/Caru19}, a relation is instead called an \emph{information set}.}.
From any tuple system, we can generate an ordered information algebra of relations in a canonical way.

\begin{theorem}[{\cite[Theorem 6.10]{DBLP:books/daglib/0008195}}]
    \label{thm.tuple_to_info}
    Let $(\fun{T}, \cat{L})$ be a tuple system.
    Define a \defn{relation over $U \in \cat{L}$} to be a subset $R \subseteq \fun{T}$ such that $\dd x = U$ for all $x \in R$, and define the \defn{domain of $R$} as $\dd R \defeq U$.
    For $U \leq \dd R$, the projection of $R$ onto $U$ is defined
    \begin{equation}
        \proj{R}{U} \defeq \setc{ \projtup{x}{U} \in \fun{T}}{ x \in R }
    \end{equation}
    For relations $R, S$ define the \defn{join of $R$ and $S$} as
    \begin{equation}
        R \otimes S
        \defeq
        \setc{
            x \in \fun{T}
        }{
            \dd x = \dd R \lor \dd S,
            \projtup{x}{ \dd R} \in R,
            \projtup{x}{ \dd S} \in S
        }
    \end{equation}

    For each $U \in \cat{L}$, define $0_U \defeq \emptyset$, called the \defn{empty relation on $U$}, and $1_U \defeq \fun{T}_U$, called the \defn{universal relation on $U$}.

    Then the set $\cat{R}_\fun{T} \defeq \coprod_{U \in \cat{L}} \fun{P} (\fun{T}_U)$ of all relations, where $\fun{P}$ is the (covariant) powerset functor, is an ordered information algebra, with ordering given by subset inclusion $\subseteq$, with null elements $0_U$ and neutral elements $1_U$, for all $U \in \cat{L}$.
\end{theorem}

We associate to the lattice $(\lambdaalg, \preceq)$ of specifications, the ordered information algebra $(\chaosalg,\cat{D},\dd, \projop, \otimes,\subseteq)$, whose valuations represent \emph{nondeterministic computations}; the nondeterminism corresponding to the multiplicity of traces in its relations.
On each domain $U \in \cat{D}$, the ordering $\subseteq$ on $(\chaosalg)U $ encodes \emph{implementability} (or \emph{refinement}) via reduction of nondeterminism, i.e. $R \subseteq S$ if and only if every trace of $R$ is also a trace of $S$.
The top element $1_U$ consists of all possible traces on $U$, whereas the bottom element $0_U$ is an empty set of traces.

Often, the combination operation of an information algebra has a canonical description via an adjunction~\cite{MR4024458}.
It is convenient to note that this holds for $(\chaosalg,\cat{D},\dd, \projop, \otimes,\subseteq)$.

\subsubsection{Adjoint structure.}

The following definition is adapted from~\cite{MR4024458} to an arbitrary finite base space.
Let $(\phialg, \cat{T}, \dd, \projop, \otimes, \leq)$ be an ordered information algebra.
Due to the universal property of products in the category $\cat{Set}$, we have, for all opens $U, W \in \cat{T}$, the following commutative diagram,
\[\begin{tikzcd}[sep=huge]
        & {\phialg(U \cup W)} \\
        {\phialg U} & {\phialg U \times \phialg W} & {\phialg W}
        \arrow[from=2-2, to=2-1]
        \arrow[from=2-2, to=2-3]
        \arrow["{\rho_U^{U \cup W}}"', from=1-2, to=2-1]
        \arrow["{\rho_W^{U \cup W}}", from=1-2, to=2-3]
        \arrow["{(\rho_U^{U \cup W}, \rho_W^{U \cup W})}"{description}, dashed, from=1-2, to=2-2]
    \end{tikzcd}\]
where $\phialg$ is viewed as a prealgebra, and where $\rho^U_W : \phialg U \to \phialg W$ are the restriction maps $x \mapsto \restrict{x}{W}$ for all $U, W \in \cat{T}$ with $W \subseteq U$.

\begin{definition}[adjoint information algebra]
    An \defn{adjoint information algebra} is an ordered information algebra $(\phialg, \cat{T}, \dd, \projop, \otimes, \leq)$ such that each restriction of its combination operation $- \otimes - : \phialg U \times \phialg W \to \phialg (U \cup W)$ is the right adjoint of the map $(\rho_U^{U \cup W},\rho_W^{U \cup W})$, defined in the diagram above.
    Hence, $\otimes$ is the unique map such that both,
    \begin{eqnarray}
        \label{eq.adjoint1}
        1_{\phialg (U \cup W)} \leq \otimes \circ (\rho_U^{U \cup W},\rho_W^{U \cup W})\\
        \label{eq.adjoint2}
        (\rho_U^{U \cup W},\rho_W^{U \cup W}) \circ \otimes \leq
        1_{\phialg U \times \phialg W}
    \end{eqnarray}
    where $\leq$ is the pointwise order induced from the partial order of the algebra, and $1_{\phialg U} : \phialg U \to \phialg U$ is the identity function on $\phialg U$ for each $U \in \cat{T}$.
\end{definition}

In other words, in an adjoint information algebra $\phialg$ with $U, W \in \cat{T}$, \cref{eq.adjoint1} says for all $\phi \in \phialg_{U \cup W}$, it holds
\begin{equation}
    \label{eq.adj_comb_right}
    \phi \leq \proj{\phi}{U} \otimes \proj{\phi}{W}
\end{equation}
and \cref{eq.adjoint2} says for all $\phi \in \phialg_U$ and $\psi \in \phialg_W$, both the following inequalities hold
\begin{equation}
    \label{eq.adj_comb_left}
    \proj{(\phi \otimes \psi)}{U} \leq \phi,\qquad
    \proj{(\phi \otimes \psi)}{W} \leq \psi
\end{equation}

\begin{theorem}\ifarxiv [proof in~\cref{thm.adjoint_info_alg_proof}]\fi
    \label{thm.adjoint_info_alg}
    An information algebra of relations $\cat{R}_\fun{T}$ over a tuple system $\fun{T}$ is adjoint.
\end{theorem}

\begin{corollary}
    The information algebra $\chaosalg$ is adjoint.
\end{corollary}

\section{Local and global consistency}

\label{sec.consistency}

In this subsection, we introduce two\footnote{A third notion of \emph{complete disagreement} is introduced in~\cite{MR4024458,DBLP:phd/ethos/Caru19}, but we do not make use of it here.} concepts of \emph{agreement} that have an interesting interpretation for specifications in $(\lambdaalg,\preceq)$.

Let $(\phialg, \cat{T}, \dd, \projop, \otimes)$ be an information algebra over a space $(X,\cat{T})$.
A finite set of valuations $K \defeq \set{\phi_1, \ldots, \phi_n} \subseteq \phialg$ is called a \defn{knowledgebase (on $\phialg$)}.
We are often interested in the case where $\bigcup_{\phi \in K} \dd \phi = X$.
\begin{definition}[local agreement]
    Two valuations $\phi, \psi \in \phialg$ \defn{locally agree} if and only if
    \begin{equation}
        \proj{\phi}{ \dd \phi \cap \dd \psi} = \proj{\psi}{ \dd \phi \cap \dd \psi}
    \end{equation}
    A knowledgebase $K = \set{\phi_1, \ldots, \phi_n} \subseteq \phialg$ \defn{locally agrees} if and only if every pair $\phi_i,\phi_j$ in $K$ locally agrees.
\end{definition}

\begin{definition}[global agreement]
    \label{def.global_agreement}
    A knowledgebase $K = \set{\phi_1, \ldots, \phi_n} \subseteq \phialg$ \defn{globally agrees} if and only if there exists\footnote{Unlike in the definition of a \emph{sheaf}, which is a presheaf on a topological space satisfying a certain continuity condition, there is no requirement that the amalgamation of local data (here $\gamma$) should be unique.
        Actually, it is common in physical applications that global sections are not unique; see for example~\cite{MR3157249} for applications of sheaf theory to the field of signal processing, where this is generally the case.
    } a valuation $\gamma \in \phialg_U$, where $U = \bigvee_{i=1}^n \dd \phi_i$ for which, for all $1 \leq i \leq n$,
    \begin{equation}
        \proj{\gamma}{ \dd \phi_i} = \phi_i
    \end{equation}
\end{definition}
The $\gamma$ of~\cref{def.global_agreement} is called a \defn{truth valuation for $K$}.

\note{Global agreement implies local agreement:
    for any pair $\phi, \psi$ in a globally agreeing knowledgebase $K$, we have
    \begin{equation}
        \proj{\phi}{\dd \phi \cap \dd \psi}
        = \proj{(\proj{\gamma}{\dd \phi })}{\dd \phi \cap \dd \psi}
        \stackrel{\text{(I3)}}{=} \proj{\gamma}{\dd \phi \cap \dd \psi}
        \stackrel{\text{(I3)}}{=} \proj{(\proj{\gamma}{\dd \psi})}{\dd \phi \cap \dd \psi}
        = \proj{\psi}{\dd \phi \cap \dd \psi}
    \end{equation}
    The converse is generally false, as we see in~\cref{ex.dining} below.
}

To any specification $(\fun{A}, \cat{U})$ we associate a knowledgebase on $\chaosalg$,
\begin{equation}
    K_{(\fun{A}, \cat{U})} \defeq
    \set{\fun{A} U}_{U \in \cat{U}}
\end{equation}
where $\dd ( \fun{A} U) = U$ for each $U \in \cat{U}$.

\begin{definition}[local/global consistency]
    The specification $(\fun{A}, \cat{U})$ is \defn{locally consistent} if and only if the associated knowledgebase $K_{(\fun{A}, \cat{U})}$ locally agrees.
    The specification $(\fun{A}, \cat{U})$ is \defn{globally consistent} if and only if $K_{(\fun{A}, \cat{U})}$ globally agrees, and the corresponding truth valuation $\gamma \in \chaosalg$ is a section of $\fun{A}$.
\end{definition}

Local consistency of a specification is a basic prerequisite for correctness.
Global consistency is a subtler correctness criterion, and is related to Lamport's definition of \emph{sequential consistency} for concurrent programs~\cite{DBLP:journals/tc/Lamport79}:

\blockquote{\ldots the result of any execution is the same as if the operations of all the processors were executed in some sequential order, and the operations of each individual processor appear in this sequence in the order specified by its program.}

Indeed, a globally consistent specification is one that can be represented by a subset of execution traces on the union of the domains of all the valuations, each one encoding a sequential ordering of states, such that when restricted to an individual domain in the context, the states occur in the same order as specified by the valuation on that domain.

The following characterises local consistency of a specification in terms of a property of the associated subpresheaf of chaos, and suggests a convenient approach to its verification.

\begin{theorem}\ifarxiv [proof in~\cref{thm.compatible_proof}] \fi
    \label{thm.compatible}
    A specification $(\fun{A}, \cat{U})$ is locally consistent if $\fun{A}$ is \defn{flasque beneath the cover $\cat{U}$}, i.e. if every restriction map $A W' \to AW$ is surjective, whenever $W \subseteq W' \subseteq U$ for some $U \in \cat{U}$.
\end{theorem}

The next result shows that in the case of an adjoint information algebra, a global valuation must take on the particular form of a solution to a so-called \emph{inference problem}~\cite{DBLP:books/daglib/0008195,pouly2012generic}, and thereby suggests a method to determine global consistency for specifications in $\lambdaalg$.

\begin{theorem}\ifarxiv [proof in~\cref{thm.global_proof}] \fi
    \label{thm.global}
    Let $\phialg$ be an adjoint information algebra,
    let $K = \set{\phi_1, \ldots, \phi_n} \subseteq \phialg$ be a knowledgebase, and let $\gamma = \bigotimes_{i=1}^n \phi_i$.
    Then $K$ agrees globally if and only if $\proj{\gamma}{\dd \phi_i} = \phi_i$ for all $1 \leq i \leq n$.
    In this case, $\gamma$ is the greatest truth valuation for $K$.
\end{theorem}

Determining if a knowledgebase is locally consistent is computationally straightforward.
Global consistency, however, is computationally intensive to verify.
To give an indication of the computational cost, assume for simplicity that for each variable $v \in V$, $\om_v = \omega$ is a constant value.
Let $K = \set{\phi_1, \ldots, \phi_n}$ be a knowledgebase.
Then to determine if $K$ is globally consistent, according to~\cref{thm.global} we must compute $\proj{(\otimes K)}{\dd \phi_j}$ for each $1 \leq j \leq n$.
To compute the join $\otimes K$ involves \enquote{filtering} from the valuations on the cartesian product of the state spaces $\om_{\dd \phi_i} \isoto \omega^{\dd \phi_i}$, i.e. the proset
\begin{equation}
    \label{eq.expensive}
    \prod_{i} \omega^{\dd \phi_i}
    \isoto
    \omega^{ \coprod_i \dd \phi_i}
\end{equation}
whose underlying set has cardinality exponential in the number of variables, and is generally intractable to compute in practice.

Fortunately, by applying the combination axiom (I6) of~\cref{axiom.combination} inductively, we can avoid computing the join $\otimes K$ directly, and instead compute for each $j$,
\begin{equation}
    \label{eq.cheaper}
    \phi_j \otimes \paren{\bigotimes_{i \neq j} \proj{\phi_i}{ \dd \phi_i \cap \dd \phi_j }}
\end{equation}
which is still exponential in the variables, but the number of variables in the exponent has been reduced, often significantly.

This is the starting point for \emph{local computation} algorithms, such as the \emph{fusion} and \emph{collect} algorithms, which are generic algorithms for computing global agreement in information algebras, which in some applications are best-in-class~\cite{DBLP:books/daglib/0008195,DBLP:phd/ethos/Caru19}.

\section{Example: the dining philosophers}
\label{sec.dining}
In~\cite{MR4024458}, a knowledgebase that locally agrees but globally disagrees is called \emph{contextual}.
We next give an example of this phenomenon---a locally consistent but globally inconsistent specification---in a classical scenario in concurrency, the \enquote{dining philosophers}.

\begin{example}
    \label{ex.dining}

    This example models a group of philosophers sat at a circular table wanting to eat a meal, with one chopstick on the table between each adjacent pair of philosophers.
    A philosopher can either think or eat.
    To eat, a philosopher must hold both their adjacent chopsticks.
    Our presentation here is based on the one in~\cite{DBLP:journals/mscs/Goguen92}.

    Let $n \geq 2$, let $\set{p_0,\ldots,p_{n-1}}$ be variables corresponding to the philosophers, and let $\set{c_0,\ldots,c_{n-1}}$ be variables corresponding to the chopsticks.
    For each $0 \leq i < n$, let
    \begin{equation}
        \om_{p_i} \defeq \set{t,e}, \qquad
        \om_{c_i} \defeq \set{i-1,*,i}
    \end{equation}
    where $t$ and $e$ stand respectively for \enquote{thinking} and \enquote{eating}, and $i-1, i$ refer respectively to the philosophers $p_{i-1}, p_i$ who may hold chopstick $c_i$, and $*$ to the neutral state of the chopstick on the table (all indices taken $\modd{n}$).

    Define a context $\cat{U} = \set{U_i}_{i =0}^{n-1}$ where $U_i \defeq \set{c_i, p_i, c_{i + 1}}$ represents the frame of reference of the philosopher $i$ as an independent asynchronous process in the distributed system.

    For example, if $n = 3$, we have
    \begin{equation}
        \cat{U} =
        \set{
            \set{c_0, p_0, c_1},
            \set{c_1, p_1, c_2},
            \set{c_2, p_2, c_0}
        }
    \end{equation}
    Let $(V \defeq \cup \cat{U},\cat{D})$ be the topological space generated by the \emph{subbasis} $\cat{U}$; i.e. $\cat{D}$ consists of all unions of intersections of elements of $\cat{U}$.

    (A visual representation of the situation is given by the \emph{\v{C}ech nerve} of the context $\cat{U}$; this is a simplicial complex whose $n$-cells are nonempty $n$-fold intersections of the $U_i$ with distinct indices.
    In the case that $n=3$, the \v{C}ech nerve of $\cat{U}$ is (the boundary of) a triangle, with only $0$-cells and $1$-cells (see~\cref{ch1.fig.dining_nerve}).)

    \begin{figure}
        \begin{equation}\begin{tikzcd}
                {U_2} && {U_1} \\
                & {U_0}
                \arrow["{\set{c_0}}", no head, from=2-2, to=1-1]
                \arrow["{\set{c_2}}", no head, from=1-1, to=1-3]
                \arrow["{\set{c_1}}", no head, from=1-3, to=2-2]
            \end{tikzcd}\end{equation}
        \caption{\v{C}ech nerve of the context $\cat{U}$ for $n=3$.}
        \label{ch1.fig.dining_nerve}
    \end{figure}

    Informally, let $(\fun{A}, \cat{U})$ be the specification containing all traces according to the following protocol: the legal state transitions on the $\set{c_i}$ for $x \in \om_{c_i}$ are:
    \begin{align}
        * & \mapsto x \\
        x & \mapsto * \\
        x & \mapsto x
    \end{align}
    meaning, a chopstick may either be picked up or put down, or remain in its current state.
    The legal state transitions on the $U_i$ are:
    \begin{align}
        \tag{rule 1, initial state} & \leadsto (*,t,*)                      \\
        \tag{rule 2} (l,x,r)        & \leadsto (l',x,r') & l,l',r,r' \neq i \\
        \tag{rule 3} (l,t,r)        & \leadsto (l',e,r') & l,l',r,r' \neq i \\
        \tag{rule 4} (l,e,*)        & \leadsto (l',e,i)  & l,l' \neq i      \\
        \tag{rule 5} (l,e,i)        & \leadsto (l',e,i)  & l,l' \neq i      \\
        \tag{rule 6} (*,e,i)        & \leadsto (i,e,i)                      \\
        \tag{rule 7} (i,e,i)        & \leadsto (*,t,*)
    \end{align}
    The first rule says philosophers begin thinking, without chopsticks.
    The second rule says if a philosopher has no chopsticks, they may continue in their present state, without constraining the actions of their two neighbours.
    The third rule says if they have no chopsticks and are thinking, then they may become hungry, without constraining their neighbours.
    The fourth rule says that if they have no chopsticks, they may pick up their right one, if it is available, without constraining their left neighbour.
    The fifth rule says they can remain in the state of having just a right chopstick, without constraining their left neighbour.
    The sixth rule says if they have a right chopstick, they may pick up the left one if it is available.
    The last rule says that if they have both chopsticks and are eating, they can put them both down and think.

    Consider a sub-specification $(\fun{B}, \cat{U}) \preceq (\fun{A}, \cat{U})$ whose corresponding knowledgebase is $K_{(\fun{B}, \cat{U})} \defeq \set{\phi_0, \phi_1, \phi_2}$, where for $0 \leq i \leq 2$, each $\phi_i$ is a singleton
    \begin{equation}
        \phi_i =
        \set{
            \begin{bmatrix}
                * & * & * & (i-1) & * & i & * & *     & * \\
                t & e & e & e     & e & e & t & t     & t \\
                * & * & i & i     & i & i & * & (i+1) & *
            \end{bmatrix}
        }
    \end{equation}
    using the matrix representation of \cref{note.matrix}, where the first row corresponds to the variable $c_i$, the second to $p_i$, and the third to $c_{i+1}$.
    For example, the command $\phi_1$ contains the single trace corresponding to the following linear sequence of events:
    \begin{enumerate}
        \item $p_1$ becomes hungry (rule 3);
        \item $p_1$ picks up the right chopstick (rule 4);
        \item $p_0$ picks up the left chopstick (rule 5);
        \item $p_0$ puts down the left chopstick (rule 5);
        \item $p_1$ picks up the left chopstick (rule 6);
        \item $p_1$ eats and puts down both chopsticks (rule 7);
        \item $p_2$ picks up the right chopstick (rule 2).
        \item $p_2$ puts down the right chopstick (rule 2).
    \end{enumerate}

    Clearly, the specification is legal according to the protocol described above, and enables each philosopher to eat their meal.
    Moreover, the specification is locally consistent; we have for each $i$,
    \begin{equation}
        \proj{\phi_i}{\set{c_i}}
        =\set{
            \begin{bmatrix}
                * & (i-1) & * & i & *
            \end{bmatrix}}
        =
        \proj{\phi_{i-1}}{\set{c_i}}
    \end{equation}

    On the other hand, it is intuitively clear that the specification cannot be globally consistent, because
    \begin{enumerate}
        \item $\phi_0$ says $p_0$ picks up $c_1$ before $p_2$ picks up $c_0$,
        \item $\phi_2$ says $p_2$ picks up $c_0$ before $p_1$ picks up $c_2$,
        \item $\phi_1$ says $p_1$ picks up $c_2$ before $p_0$ picks up $c_1$,
    \end{enumerate}
    and together these events form a causal loop, which is physically impossible, and moreover, not representable as a trace on $V = U_0 \cup U_1 \cup U_2$.
    This can be calculated formally using~\cref{eq.cheaper}, but we omit the details for reasons of space.

    This example illustrates that global consistency of a specification is an important criterion for correctness.
\end{example}

\section{Conclusion}

We have presented a refinement lattice of specifications to model distributed programs, using mathematical structures that emphasise the intrinsic topological structure of the distributed system.
The specifications in our lattice consist of subpresheaves of \emph{relative} traces, for which the absolute timing of events is not preserved under restriction maps, but only their relative ordering.
This aspect was emphasised to reflect fundamental physical constraints on synchronisation---at high speeds, such as those of modern computer technology, Einstein has taught us that the idea of synchronous events loses its meaning.
This structure of relative traces then revealed an interesting correctness criterion for specifications, related to Lamport's definition of \emph{sequential consistency}.

\subsubsection{Acknowledgements.}

Nasos is grateful to his PhD advisor Larissa Meinicke for helpful feedback; to the Category Theory Zulip community\footnote{\url{https://categorytheory.zulipchat.com}} for many helpful conversations, and especially to Amar Hadzihasanovic who suggested the proof of Lemma 3, and also for the support of the Australian Government Research Training Program Scholarship.
This research was supported by Discovery Grant DP190102142 from the Australian Research Council (ARC).
We thank the anonymous reviewers for their helpful comments and suggestions.
\ifarxiv
    \appendix

\section{Proofs}
\label{app.proofs}

\subsection{Proof of~\cref{lem.nondegen_functor}}
\label{lem.nondegen_functor_proof}

\begin{proof}
    Clearly $\fun{D}1_X = 1_{\fun{D}X}$.

    Let $\alpha : S \to S'$, $\beta : S'\to S''$ be morphisms of augmented simplicial sets and let $x \in \fun{D} S$.
    Let $z \defeq (\fun{D}\alpha)x$, $y \defeq (\fun{D}(\beta \circ \alpha))x$, and $y' \defeq (\fun{D} \beta)z$.
    Then by the definition of $\fun{D}$ and~\cref{ch1.lem:eilenberg-zilber}, there exist unique $f,g,h \in \deltad$ such that
    \begin{align}
        \contra{f} y  & = \beta(\alpha x) \\
        \contra{g} z  & = \alpha x        \\
        \contra{h} y' & = \beta z
    \end{align}
    Now we have
    \begin{align}
        \contra{f} y
         & = \beta(\alpha x)            \\
         & = \beta(\contra{g} z)        \\
         & = \contra{g} (\beta z)       \\
         & = \contra{g} (\contra{h} y') \\
         & = \contra{(h \circ g)} y'
    \end{align}
    where we used naturality of $g$.
    Since $h \circ g$ is surjective (as a composite of surjections), by uniqueness we conclude $y = y'$.
    Since $x$ was arbitrary, this verifies that $\fun{D}$ is a functor.
\end{proof}

\subsection{Proof of~\cref{thm.covers}}
\begin{proof}
    \label{thm.covers_proof}
    Antichains on a finite poset are well known\footnote{This is a straightforward consequence of the \emph{fundamental theorem of distributive lattices}~\cite[pp. 104–112]{MR2868112}.} to form a complete distributive lattice with the operations~\cref{eq.antichain_meet,eq.antichain_join}, so if $\maxcov{\cat{D}}$ is a lattice under the ordering it must have the same meet and join, and we only need to check it is closed under those operations.

    If $\cat{U},\cat{W}$ are maximal covers and $x \in X$ we know there is a $U \in \cat{U}, V \in \cat{W}$ with $x \in U$ and $x \in V$.

    For the join, if $U,V$ are incomparable then $U \in \cat{U} \lor \cat{W}$ and so $x \in \bigcup{(\cat{U} \lor \cat{W})}$, so we have $ X = \bigcup{(\cat{U} \lor \cat{W})}$ since $x$ was arbitrary.
    If $U, V$ are comparable, then the larger one is in $\cat{U} \lor \cat{W}$ and again $ X = \bigcup{(\cat{U} \lor \cat{W})}$.

    For the meet, we have $x \in U \cap V$ and $U \cap V \in L_\cat{U} \cap L_\cat{W}$.
    If there is no $Z \supset U \cap V$ in $L_\cat{U} \cap L_\cat{W}$, then we have $U \cap V \in \cat{U} \land \cat{W}$ and therefore $x \in \bigcup{(\cat{U} \land \cat{W})}$ and $X = \bigcup{(\cat{U} \land \cat{W})}$.
    Otherwise, let $Z$ be the maximal element of $L_\cat{U} \cap L_\cat{W}$ above $U \cap V$, which must exist because $L_\cat{U} \cap L_\cat{W}$ is a finite set. Then $Z \in \cat{U} \land \cat{W}$ and so $x \in \bigcup{(\cat{U} \land \cat{W})}$ and $X = \bigcup{(\cat{U} \land \cat{W})}$.
\end{proof}

\subsection{Proof of~\cref{thm.concurrency_tuple}}
\label{thm.concurrency_tuple_proof}

\begin{lemma}
    \label{lem.degen_extension}
    Let $\deltad$ be the subcategory of $\deltaplus$ with only surjections; equivalently, the category whose
    \begin{itemize}
        \item objects are finite linear orders $\lin{n} = \set{0 \leq \cdots \leq n}$ for $n \in \Nn$, and
        \item morphisms are surjective monotone functions.
    \end{itemize}
    For each $n$, consider the poset $\cat{Deg}_n$ whose
    \begin{itemize}
        \item objects are functions $f : \lin{m} \to \lin{n}$ in $\deltad$ with target $\lin{n}$, and
        \item for $f, g \in \cat{Deg}_n$, $f \leq g$ if and only if $f$ factors through $g$, i.e. $f = g \circ h$ for some surjection $h \in \deltaplus$.
              (This is the poset reflection of the slice category $\deltad/\lin{n}$.)

              Then for all $n$, the poset $\cat{Deg}_n$ has binary meets; that is, for all $f$ and $g$, there is a \emph{greatest} $f \land g$ factoring through $f$ and $g$, in the sense that
              \begin{enumerate}
                  \item $f \land g \leq f$ and $f \land g \leq g$, and
                  \item for any $h$ with $h \leq f$ and $h \leq g$ we also have $h \leq f \land g$.
              \end{enumerate}
    \end{itemize}
\end{lemma}
\begin{proof}
    Any $f : \lin{m} \to \lin{n}$ can be represented as a tuple $(f_0, \ldots, f_n)$ where $f$ sends the first $f_0$ elements to $0$, the next $f_1$ to $1$, etc., and where $\sum_{i=0}^n f_i = m + 1$.
    (For example, the function $f : \lin{3} \to \lin{1}$ given by $(0,1,2,3) \mapsto (0,0,0,1)$ is represented by the tuple $(3,1)$.)

    Given $f: \lin{m_1} \to \lin{n}$ and $g: \lin{m_2} \to \lin{n}$, represent them as tuples $(f_0, \ldots, f_n)$ and $(g_0, \ldots, g_n)$ in this way.

    Let $f \land g: \lin{m} \to \lin{n}$ correspond to the tuple $(\max(f_0, g_0), \ldots, \max(f_n, g_n))$, where $m \defeq -1 + \sum_i \max(f_i, g_i)$.

    Then $f \land g$ factors through $f$ (possibly non-uniquely): take any list of $n+1$ surjective functions
    \begin{equation}
        f'_i :
        \lin{\max(f_i, g_i) - 1 } \to \lin{f_i - 1}
    \end{equation}
    for $0 \leq i \leq n$, and concatenate their tuple representations together to get a surjective function $f': \lin{m} \to \lin{m_1}$ that is easily verified to satisfy $f \land g = f \circ f'$.
    Similarly, $f \land g$ factors through $g$ (possibly non-uniquely).

    Let $h: \lin{m'} \to \lin{n}$ be any other map that factors through both $f$ and $g$, and is represented by the tuple $(h_0, \ldots, h_n)$.
    Then it is straightforward to see that $h_i \geq f_i$ and $h_i \geq g_i$ for all $0 \leq i \leq n$ (this is essentially because the \enquote{size} of inverse images of an element can only increase by precomposition with a surjective map).

    It follows that $h_i \geq \max(f_i, g_i)$ for all $i$. Then taking any list of $n+1$ surjective maps $\lin{h_i - 1} \to \lin{\max(f_i, g_i) - 1}$ and concatenating them together, we get a map $k: \lin{m'} \to \lin{m}$ for which $h = (f \land g) \circ k$, and this completes the proof.
\end{proof}

So what~\cref{lem.degen_extension} implies is that if $f^* z$ and $g^* z$ are degeneracies of the same nondegenerate cell $z$, ${(f \land g)}^* z$ is the “smallest unifier” in this sense: if $h^* z$ is a degeneracy of both $f^* z$ and $g^* z$, then it is a degeneracy of ${(f \land g)}^* z$.

\begin{lemma}
    \label{lem.chaos_flasque}
    The functor $\chaos$ is \emph{flasque}, i.e. for all $U, W \in \cat{D}$ with $U \subseteq W$, the map $\chaos W \to \chaos U$ is surjective.
\end{lemma}
\begin{proof}
    Assume $x$ is a trace in $\chaos U$ of length $n$ and let $U \subseteq W$, for $U, W \in \cat{D}$.
    Let $y \in (\nerve \circ \om) W $ be an $n$-cell such that $\rho^W_U{y} = x$, and outside $U$, $y$ takes any values whatsoever---such a $y$ clearly exists because by assumption $\Omega_v$ is nonempty for each $v \in \dd y$.
    Note that restrictions commute with degeneracy maps by naturality, i.e. for any $f \in \deltad$ and $x \in (\nerve \circ \om)W$,
    \begin{equation}
        \contra{f} (\rho^W_U x)
        = \rho^W_U{(\contra{f}x)}
    \end{equation}
    where $\rho^W_U : (\nerve \circ \om)W \to (\nerve \circ \om)U$ denotes the restriction map corresponding to the functor $\nerve \circ \om$.
    Since $x = \rho^W_U y$ is nondegenerate by assumption, it then follows by contraposition that $y$ is nondegenerate, and so $y \in \chaos W$ and $\restrict{y}{U} = x$.
    Since $U, W$ are arbitrary, this shows that $\chaos$ is flasque.
\end{proof}

\begin{proof}[\cref{thm.concurrency_tuple}]
    Axioms (T1)--(T3) are immediate from functoriality of $\chaos$, and (T5) is clearly equivalent to~\cref{lem.chaos_flasque}, so it remains only to show (T4).

    Let $U \defeq \dd x$, $W \defeq \dd y$, and $w \defeq \restrict{x}{U \cap W} = \restrict{y}{U \cap W} $.

    First assume $U \cap W = \emptyset$.
    Then we necessarily have $w = \unit$.
    Let $n_x$ be the length of $x$ and $n_y$ the length of $y$.
    If $n_x = n_y$, then $z \defeq (x,y) \in \chaos (U \sqcup W)$ is clearly a nondegenerate cell that satisfies the condition for (T4).
    Otherwise, assume without loss of generality $n_x < n_y$.
    Then letting $z \defeq (\contra{f} x, y)$ where $f$ is any surjection $\lin{n_y} \to \lin{n_x}$, we have again that $z$ is a nondegenerate cell satisfying (T4).

    Now assume $U \cap W \neq \emptyset$.
    Let $\rho^X_Y : (\nerve \circ \om)X \to (\nerve \circ \om)Y$ denote the restriction maps corresponding to the functor $\nerve \circ \om$ for any $X, Y \in \cat{D}$.
    Now let $w_x \defeq \rho^U_{U \cap W}x$ be the (possibly degenerate) restriction of $x$ with respect to the functor $\nerve \circ \om$, and $w_y \defeq \rho^W_{U \cap W} y$ the corresponding (possibly degenerate) restriction of $y$, so there exist unique surjections $f_x,f_y$ with
    \begin{equation}
        w_x = \contra{f_x} w,\qquad
        w_y = \contra{f_y} w
    \end{equation}
    Let $f_x \land f_y$ be the minimal extension of $f_x$ and $f_y$ as given by~\cref{lem.degen_extension}, so that there exists (possibly not unique) surjections $h_x$ and $h_y$ with
    \begin{equation}
        \contra{h_x} w_x
        = \contra{(f_x \land f_y)} w
        = \contra{h_y} w_y
    \end{equation}

    Now $(\contra{h_x} x, \contra{h_y} y)$ is a cell in $\sub{(\nerve \circ \om)U}{m} \times \sub{(\nerve \circ \om)W}{m}$ for some $m$, that when viewed as a matrix, has the sequence of rows corresponding to $U \cap V$ of both factors coinciding:
    \begin{align}
        \rho^U_{U \cap W}(\contra{h_x} x)
         & = \contra{h_x} (\rho^U_{U \cap W} x) \\
         & = \contra{h_x} w_x                   \\
         & = \contra{(f_x \land f_y)} w         \\
         & = \contra{h_y} w_y                   \\
         & = \contra{h_y} (\rho^W_{U \cap W} y) \\
         & = \rho^W_{U \cap W}(\contra{h_y} y)
    \end{align}
    by naturality.
    By identifying these equal rows, we obtain a cell $z \in \sub{(\nerve \circ \om)(U \cup W)}{m}$ in an evident way.

    Assume there is a surjection $e \in \deltad$ with $z = \contra{e} z'$.
    Then $e$ is a common degeneracy of both $\contra{h_x} x$ and $\contra{h_y} y$,
    and hence of both $\contra{h_x} w_x$ and $\contra{h_y} w_y$, and so there exist unique surjections $q_x, q_y \in \deltad$ for which
    \begin{equation}
        (\contra{h_x} w_x, \contra{h_y} w_y)
        = \contra{e} (\contra{q_x} w_x, \contra{q_y} w_y)
        = (\contra{e} (\contra{q_x} w_x), \contra{e} (\contra{q_y} w_y))
    \end{equation}
    Hence we have that
    \begin{equation}
        f_x \land f_y
        = f_x \circ q_x \circ e
        = f_y \circ q_y \circ e
    \end{equation}
    Because $e$ is surjective, it is right-cancellative, and we conclude
    \begin{equation}
        f_x \circ q_x = f_y \circ q_y \eqdef \bang
    \end{equation}
    as in the below diagram in $\deltad$,
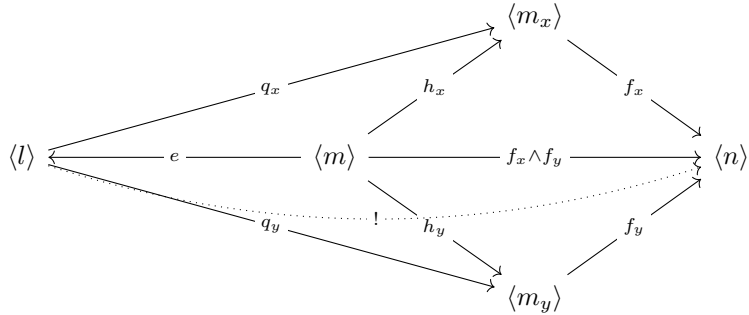
    \[\begin{tikzcd}[sep=huge]
            &&& {\lin{m_x}} \\
            {\lin{l}} && {\lin{m}} && {\lin{n}} \\
            &&& {\lin{m_y}}
            \arrow["{f_x\land f_y}"{description}, from=2-3, to=2-5]
            \arrow["{f_x}"{description}, from=1-4, to=2-5]
            \arrow["{f_y}"{description}, from=3-4, to=2-5]
            \arrow["{h_x}"{description}, from=2-3, to=1-4]
            \arrow["{h_y}"{description}, from=2-3, to=3-4]
            \arrow["{q_x}"{description}, from=2-1, to=1-4]
            \arrow["e"{description}, from=2-3, to=2-1]
            \arrow["{q_y}"{description}, from=2-1, to=3-4]
            \arrow["{!}"{description}, curve={height=30pt}, dotted, from=2-1, to=2-5]
        \end{tikzcd}\]
    for some $l$.

    Now we have
    \begin{equation}
        f_x \land f_y \leq \bang,\qquad
        \bang \leq f_x, \qquad ! \leq f_y
    \end{equation}
    But since $\bang \leq f_x \land f_y$, we conclude $ f_x \land f_y = \bang$, so that $e = 1$ is the identity function, and $z=z'$ is nondegenerate.
    Moreover, $\restrict{z}{U} = x$ and $\restrict{z}{W} = y$, so (T4) is verified.
\end{proof}

\subsection{Proof of~\cref{thm.adjoint_info_alg}}
\label{thm.adjoint_info_alg_proof}

This proof is essentially the same as~\cite[Proposition A.1.]{MR4024458}.
We reproduce it here for convenience, and to show it does not depend on discreteness of the base space.
\begin{proof}
    Let $R \subseteq \fun{T}_{U \lor W}$ where $U, W \in \cat{L}$.
    We have
    \begin{align}
        \proj{R}{U} \otimes \proj{R}{W}
         & = \setc{\projtup{x}{U}}{x \in R}
        \otimes \setc{\projtup{x}{W}}{x \in R} \\
         & = \setc{
            z \in \fun{T}_{U \lor W}
        }{
            \te{x,y \in R}{\projtup{z}{U}=\projtup{x}{U}, \projtup{z}{W}=\projtup{y}{W}}
        }
    \end{align}
    so that clearly, $R \subseteq \proj{R}{U} \otimes \proj{R}{W}$.

    Now, let $R_1 \subseteq \fun{T}_U$ and $R_2 \subseteq \fun{T}_W$.
    Then
    \begin{align}
        \proj{(R_1 \otimes R_2)}{U}
         & = \proj{
            \setc{z \in \fun{T}_{U \lor W}}{
                \projtup{z}{U} \in R_1, \projtup{z}{W} \in R_2
            }
        }{W}                         \\
         & = \setc{\projtup{x}{U} }{
            x \in \fun{T}_{U \lor W},
            \projtup{x}{U} \in R_1, \projtup{x}{W} \in R_2
        }                            \\
         & \subseteq R_1
    \end{align}
    Similarly, $\proj{(R_1 \otimes R_2)}{W} \subseteq R_2$.
\end{proof}

\subsection{Proof of~\cref{thm.compatible}}
\label{thm.compatible_proof}

This proof is based on~\cite[Prop. 6.1]{MR4024458}.
\begin{proof}
    Let $K_{(\fun{A},\cat{U})} \defeq \set{\fun{A} U}_{U \in \cat{U}}$ be the knowledgebase associated to the specification $(\fun{A}, \cat{U})$.
    Let $U,W \in \cat{U}$.
    Then $x \in \restrict{\fun{A} U}{U \cap W}$ if and only if $x \in \fun{A}(U \cap W)$, by flasqueness of $\fun{A}$.
    Similarly, $x \in \restrict{\fun{A} W}{U \cap W}$ if and only if $x \in \fun{A}(U \cap W)$.
    Hence,
    \begin{equation}
        \restrict{\fun{A} U}{U \cap W}
        = \fun{A}(U \cap W)
        = \restrict{\fun{A} W}{U \cap W}
    \end{equation}
    which means $K_{(\fun{A},\cat{U})}$ is compatible.
\end{proof}

\subsection{Proof of~\cref{thm.global}}
\label{thm.global_proof}

This proof is taken verbatim from~\cite[Proposition 5.2]{MR4024458}.
We reproduce it here for convenience, and to show it does not depend on discreteness of the base space.
\begin{proof}
    Suppose $\delta \in \phialg_V$ is a truth valuation for $K$, i.e. $\proj{\delta}{d \phi_i} = \phi_i$ for all $1 \leq i \leq n$.
    Since $\phialg$ is adjoint, we have
    \begin{equation}
        \delta
        \stackrel{(\ref{eq.adj_comb_right})}{\leq} \bigotimes_{i=1}^n \proj{\delta}{d \phi_i}
        = \bigotimes \phi_i
        = \gamma
    \end{equation}
    Moreover, because projection is monotone by axiom (O4) of~\cref{def.ordered_info_alg}, we have
    \begin{equation}
        \phi_i
        \leq \proj{\delta}{d \phi_i}
        \stackrel{\text{(O4)}}{\leq} \proj{\gamma}{d \phi_i}
        \stackrel{(\ref{eq.adj_comb_left})}{\leq} \phi_i
    \end{equation}
    So we conclude $\proj{\gamma}{\dd \phi_i} = \phi_i$ for each $1 \leq i \leq n$, and thus $\gamma$ is a truth valuation for $K$.
\end{proof}

\fi
%
%
%
\bibliographystyle{splncs04}
\bibliography{main.bib}
%





\end{document}